\documentclass[aps,pra,showpacs,twoside,twocolumn,10pt]{revtex4-1}
\usepackage[colorlinks=true, citecolor=red, urlcolor=blue ]{hyperref}
\usepackage{times,epsfig,amssymb,amsfonts,amsmath,bm,subfigure,mathtools,amsthm,braket,soul,enumitem,color,physics,graphics,graphicx}
\usepackage[normalem]{ulem}
\usepackage{comment}
\usepackage{epstopdf}
\newtheorem{theorem}{Theorem}

\newtheorem{lemma}{Lemma}

\begin{document}

\title{
Typical behaviour of genuine multimode entanglement of pure Gaussian states
}

\author{Saptarshi Roy}

\affiliation{QICI Quantum Information and Computation Initiative, Department of Computer Science,
The University of Hong Kong, Pokfulam Road, Hong Kong	
	}

\begin{abstract}
Trends of genuine entanglement in Haar uniformly generated multimode pure Gaussian states with fixed average energy per mode are explored. A distance-based metric known as the generalized geometric measure (GGM) is used to quantify genuine entanglement. The GGM of a state is defined as its minimum distance from the set of all non-genuinely entangled states. To begin with, we derive an expression for the Haar averaged value of any function defined on the set of energy-constrained states. Subsequently, we investigate states with a large number of modes and provide a closed-form expression for the Haar averaged GGM in terms of the average energy per mode.
Furthermore, we demonstrate that typical states closely approximate their Haar averaged GGM value, with deviation probabilities bounded by an exponentially suppressed limit. We then analyze the GGM content of typical states with a finite number of modes and present the distribution of GGM. Our findings indicate that as the number of modes increases, the distribution shifts towards higher entanglement values and becomes more concentrated. We quantify these features by computing the Haar averaged GGM and the standard deviation of the GGM distribution, revealing that the former increases while the latter decreases with the number of modes.
\end{abstract}	

\maketitle

\section{Introduction}
\label{sec:intro}
The study of Gaussian states enjoys a privileged position in the investigation of continuous-variable quantum systems. 
The special status is enforced by a rather rare coincidence in the world of physical theories where mathematical elegance and experimental feasibility go hand in hand \cite{Adesso2014, contvar-rev, contvar-book, ferraro}. The mathematical simplicity is revealed in the phase space formalism of quantum mechanics.  Like Gaussian functions, in the phase space, the Gaussian states can be completely characterized by the first two moments, namely the displacement vector and the covariance matrix \cite{ferraro}. 

After the characterization of Gaussian states, the description of Gaussian operations involving state transformation and measurement was translated into the phase space. A considerable research volume has been allotted for translating the usual notions of quantum correlations \cite{Adesso2014, contvar-rev, contvar-book,ferraro}, quantum protocols \cite{contvar-rev,p1,p2,p3} etc in the phase space formalism. In particular, for quantum correlations, the aim is to express them in terms of invariants under phase space (Symplectic) operations, which goes by the name of symplectic invariants.  In this work, our central quantity of interest is the genuine multiparty entanglement in multimode pure Gaussian. Apart from a mere theoretical interest, genuine multiparty entanglement is an important ingredient for implementing various quantum information protocols, especially in the network scenario \cite{ggmr1,ggmr2,ggmr3}. This motivates our investigations from a practical point of view.

Quantifying genuine entanglement in quantum systems is a significant step in analyzing multiparty quantum correlations.  Although various measures are proposed, most of them suffer from issues in computability. At this point, the Generalized geometric measure (GGM) efficiently quantifies the genuine entanglement content of pure states \cite{ggm1} (see also \cite{ggm22,ggm3,ggm4,ggm5,ggm6,ggm7,ggm8}). It is a distance-based measure, where the GGM of a state is the minimum distance of the state from the set of non-genuinely entangled states. For pure states, this geometrical minimization can be transformed into an algebraic maximization over Schmidt coefficients across all its relevant bipartition. Several follow-up works investigating genuine entanglement have also utilized GGM \cite{gs1,gs2,gs3,gs4}.  Generalizing its computation for mixed states was attempted \cite{ggm2}, but positive results were restricted to states with special symmetries. Nevertheless, GGM was successfully translated to the phase space for pure Gaussian in terms of symplectic invariants in \cite{ggmcv}. Henceforth it has been used extensively to track genuine entanglement for multimode Gaussian states in different contexts \cite{lggmcv, ggmcv1,ggmcv2,ggmcv3,ggmcv4,ggmcv5,ggmcv6}.

In this work, we analyze the genuine entanglement properties of Haar uniformly generated multimode pure Gaussian states of fixed average energy per mode. 
We provide a brief description of their generation and furnish an expression of the Haar averaged value of any function defined on the set of states.
Then we concern ourselves with results involving states with a large number of modes.  First, we note that the symplectic spectrum \cite{fukuda} of $k$-mode reductions of $n$-mode pure Gaussian states is sharply concentrated around the Haar averaged value for $n >\!> 1$. Using this feature, we infer that the  GGM can be characterized by the symplectic eigenvalues of only the single-mode reduced states. This consequently allows for a closed-form expression of the Haar averaged value GGM in terms of the average energy per mode. Finally, we show that the GGM of any typical state is very close to the Haar averaged value of GGM. The above statement is made quantitative by bounding the probability that a typical state possesses a GGM value different from the Haar averaged value with an exponentially suppressed bound. 

Then we look at the GGM content of typical states with a finite number of modes ($n = 3, 4, 5,$ and $6$).  For each $n$, we present the distribution of GGM values of the Haar uniformly generated states. We find that as the number of modes increases the distribution shifts to greater GGM values and becomes progressively sharper. Quantitatively we capture these features by computing the Haar averaged value of the GGM  and the standard deviation of the GGM distribution for each $n$. Our analyses reveal that while the Haar averaged value of the GGM increases with the number of modes, the standard deviation decreases.

The contents of the paper are laid out as follows. After a brief introduction in Sec. \ref{sec:intro}, we discuss the prerequisites to discuss the results in Sec. \ref{sec:prereq}. We discuss the technique of constructing Haar averaged quantities for typical pure Gaussian states with constrained energy per mode in Sec. \ref{sec:Haaravg}.   The GGM characteristics of a large number of modes are showcased in Sec. \ref{sec:nh}, while exact numerical computations for a few modes is presented in Sec. \ref{sec:nl}. Finally, we conclude in Sec. \ref{sec:con}. 

\section{Setting the stage}
\label{sec:prereq}
In this section, we discuss the prerequisites required to present our results.  We begin with a quick overview of the phase space formalism for Gaussian states and operations.  
This brief primer is followed by an overview of the genuine entanglement measure we use to quantify and analyze the genuine multimode entanglement of random $n$-mode pure Gaussian states.

\subsection{The phase space representation of Gaussian states and operations}
Consider a bosonic system of $n$-modes with a free Hamiltonian given  by
\begin{eqnarray}
    H = \sum_{j = 1}^n H_j, ~~\text{where} ~H_j = \frac12 (q_j^2 + p_j^2),
\end{eqnarray}
where $q_k$ and $p_l$ are the phase space quadrature operators satisfying the canonical commutation relation $[q_k,p_l] = i \delta_{kl}$.  For a compact notation, we introduce a vector of all the quadrature operators $\bm{R} = (q_1,q_2, \ldots, q_n,p_1,p_2, \ldots, p_n)^T$. The canonical commutation relations can be rewritten as
\begin{eqnarray}
    [\bm{R}_k,\bm{R}_l] = i J_{kl} \delta_{kl},
\end{eqnarray} 
where $J$, the symplectic matrix, is an antisymmetric $2n \times 2n$ matrix given 
\begin{eqnarray}
    J = \begin{bmatrix}
        \bm{0}_{n \times n} & -\mathbb{I}_{n \times n} \\
        \mathbb{I}_{n \times n} & ~~\bm{0}_{n \times n}
    \end{bmatrix}.
\end{eqnarray}
Like Gaussian functions, Gaussian states can be characterized in the phase space by the first two moments: the displacement vector $\bm{d}$ and the covariance matrix $\sigma$. 
\begin{eqnarray}
    \bm{d}_k &=& \langle \bm{R}_k \rangle, \nonumber \\
    \sigma_{kl} &=& \frac12 \langle \{ \bm{R}_k, \bm{R}_l \} \rangle - \langle  \bm{R}_k \rangle \langle \bm{R}_l \rangle.
\end{eqnarray}

Now we would shift our attention to Gaussian operations. An affine Symplectic group ${\tt ISp}(2n, \mathbb{R})$ \cite{ferraro} characterizes the most general Gaussian unitaries denoted by the pair $(\mathcal{S},\bm{v})$. Here   $\mathcal{S}\in~{\tt Sp}(2n,\mathbb{R})$ is a real $2n \times 2n $ symplectic matrix satisfying $\mathcal{S}J\mathcal{S}^T = J$, and $\bm{v}$ denotes any phase space displacements (translations).  It induces the following transformations
\begin{eqnarray}
    \bm{d} &\to& \mathcal{S}\bm{d}+\bm{v} \nonumber \\
    \sigma &\to& \mathcal{S}\sigma\mathcal{S}^T.
\end{eqnarray}
Since the correlations are entirely contained in the covariance matrices, without loss of generality, we will restrict our attention to Gaussian states with zero mean and translation-free Gaussian operations.
For Gaussian states, in particular for which $\bm{d} = 0$, we have an elegant expression for the average energy
\begin{eqnarray}
 \langle H \rangle = \frac12 \text{Tr} ~\sigma.
 \label{eq:avgen}
\end{eqnarray}
When dealing with multiple modes, not only the total energy, but the average energy per mode turns out to be an important quantity. Since the total energy is extensive, for a large number of modes $ n >\!> 1$, we can have $\langle H \rangle >\!> 1$. Therefore the quantity to look at is the average energy per mode. It is simply given by 
\begin{eqnarray}
    \frac{\langle H \rangle}{n} = \frac{1}{2n} \text{Tr} ~\sigma.
    \label{eq:avgenpm}
\end{eqnarray}

\subsection{Genuine multimode entanglement}
The genuine multimode entanglement of an $n$-mode pure Gaussian state can be computed using the generalized geometric measure (GGM) \cite{ggmcv}. It is a distance-based measure, where the GGM of any state is the minimum distance of the state from the set of non-genuinely entangled states. This optimization can be efficiently performed for pure states, where the GGM can be expressed in terms of the Schmidt coefficients of its various bipartitions. In particular, the GGM of any $n$-party pure state can be expressed as 
\begin{eqnarray}
    \mathcal{G} = 1 - \lambda_{\max}, 
    \label{eq:ggmcanonical}
\end{eqnarray}
where $\lambda_{\max}$ is the maximal Schmidt coefficient among all the Schmidt coefficients from all bipartitions of the given state. For a $n$-mode pure Gaussian state, the GGM in Eq. \eqref{eq:ggmcanonical} can be expressed  in terms  of the symplectic invariants \cite{ggmcv}
\begin{eqnarray}\label{eq:GGM-new-form}
\mathcal{G} = 1 - \max \mathcal{P}_m \Big\lbrace \prod_{i=1}^m \frac{2}{1 + \nu_i} \Big\rbrace_{m=1}^{\big[\frac{n}{2}\big]},
\end{eqnarray}
where $\mathcal{P}_m$ denotes all the reduced states with $m$-modes, and $[x]$ denotes the integral part of $x$. 

\section{Computation of Haar averages}
\label{sec:Haaravg}
Any $n$-mode pure Gaussian state can be obtained from the vacuum state by applying a Gaussian unitary. Therefore their covariance matrix can be written as $\sigma = \mathcal{S}\mathcal{S}^T$, where $\mathcal{S}$ is a symplectic operation. Using the Euler decomposition, every $\mathcal{S} \in {\tt Sp}(2n, \mathbb{R})$ can be written as
\begin{eqnarray}
    \mathcal{S} = O Z O^\prime,
  \end{eqnarray}
  where $O, O^\prime$ are orthogonal symplectic matrices, and $Z = \bm{D}\oplus \bm{D}^{-1}$ with $\bm{D}$ being a positive diagonal matrix. $Z$ corresponds to $n$ single mode squeezing unitaries. This allows us to write the covariance matrix of a $n$-mode pure Gaussian state as
  \begin{eqnarray}
      \sigma = O \Gamma O^T,
      \label{eq:random1}
  \end{eqnarray}
  where $\Gamma =Z^2$ is tensor product of $n$ single mode squeezed states. We can now express $\Gamma = \text{diag} \{z_1,z_2, \ldots, z_n\}\oplus \text{diag}\{z_1^{-1}, z_2^{-1}, \ldots, z_n^{-1} \}$, where $z_i \in [1,\infty)$. The issue with Haar uniform generation of Gaussian states lies in the fact that the group ${\tt Sp}(2n,\mathbb{R})$ is not compact \cite{fukuda}.
However, the group of orthogonal symplectic matrices ${\tt K}(n) := {\tt Sp}(2n, \mathbb{R}) \cap {\tt O}(2n)$ is isomorphic to the complex unitary group ${\tt U}(n)$, where ${\tt O}(2n)$ denotes the orthogonal group. The Haar measure on ${\tt U}(n)$ induces a Haar measure on ${\tt K}(n)$ where the isomorphism ${\tt U}(n) \to {\tt K}(n)$ is given by
  \begin{eqnarray}
      U \in {\tt U}(n) \to \begin{bmatrix}
          \text{Re}(U) & \text{Im}(U) \\
          -\text{Im}(U) & \text{Re}(U)
      \end{bmatrix} = O(U) \in {\tt K}(n).
      \label{eq:u2o}
  \end{eqnarray}
 The orthogonal symplectic matrices are passive since they keep the average energy invariant. This follows from Eq. \eqref{eq:avgen} by noting that Tr $(O \sigma O^T) =$ Tr $\sigma$ for all $O \in {\tt K}(n)$.

  Physically the Euler decomposition allows us to split the symplectic operation into an active and passive component. While the passive part supports a Haar measure, the active part does not. The reason for this is simple: \textit{unbounded squeezing}. This in turn translates to divergent energy of the generated state.   
  It can only be tamed by applying additional constraints suppressing its divergences, like an energy bound. One way to achieve this is by bounding the average energy per mode. Following Eq. \eqref{eq:avgenpm}, this translates to $\frac{1}{2n}$Tr$\Gamma \leq \bar{\nu}$, where $\bar{\nu}>1$ is a universal constant. However, since correlations are sensitive to the energy content of the state, in our work, we concentrate on Haar uniform generation of pure $n$-mode Gaussian states with a fixed average energy per mode:
  \begin{eqnarray}
      \frac{1}{2n} \text{Tr} ~\Gamma = \frac{1}{2n} \sum_{i = 1}^n \Big(z_i + \frac{1}{z_i} \Big) = \bar{\nu}.
      \label{eq:enboundpm2}
  \end{eqnarray}
  Let $\{\bm{\Gamma}\}$ be a set of all covariance matrices satisfying Eq. \eqref{eq:enboundpm2}. Finally, a pure $n$-mode Gaussian state can be generated by randomly choosing the pair $\Gamma, O$, where $\Gamma \in \{\bm{\Gamma}\}$ and $O \in {\tt K}(n)$. The covariance matrix of the randomly generated state is as mentioned in Eq. \eqref{eq:random1}, $\sigma = O \Gamma O^T$. 
  
  In general, the average value of any function $f: f(O, \Gamma) \to \mathbb{R}$ is given by a dual average: First, over the Haar measure on the orthogonal Symplectic group ${\tt K}(n)$ and, secondly over the set $\{ \bm{\Gamma}\}$.
  \begin{eqnarray}
      \mathbb{E}[f] := \Big\langle \frac{1}{V} \int f(O, \Gamma) d_{\mu}(O) \Big\rangle_{\{\bm{\Gamma}\}},
      \label{eq:hagen}
  \end{eqnarray}
  where $d_{\mu}(O)$ is invariant Haar measure on ${\tt K}(n)$, and $V = \int d_{\mu}(O)$. Here $\langle . \rangle_{\{\bm{\Gamma}\}}$ denotes averaging with respect to the set $\{\bm{\Gamma}\}$. However, we will show next that this average quantity can be expressed in a simpler form.
  \begin{theorem}
  The average value of any function $f: f(O, \Gamma) \to \mathbb{R}$ over random pure Gaussian states is given by    \begin{eqnarray}
       \mathbb{E}[f] = \frac{1}{V} \int f(O, \Gamma) d_{\mu}(O),
  \end{eqnarray}
  for any $\Gamma \in \{\bm{\Gamma}\}$.
  \label{theorem:1}
  \end{theorem}
  \begin{proof}
 Note that for any two $\Gamma, \Gamma^\prime \in \{\bm{\Gamma}\}$, energy conservation dictates that we have some passive Gaussian unitary (orthogonal symplectic transformation) connecting them: $\Gamma^\prime = O^\prime \Gamma (O^\prime)^T$, with $O^\prime \in {\tt K}(n)$. Therefore we have
  \begin{eqnarray}
      f(O,\Gamma^\prime) = f(OO^\prime, \Gamma). 
      \label{eq:equiv1}
  \end{eqnarray}
  Using the right and left invariance of the Haar measure \cite{Karol}, we get
  \begin{eqnarray}
      \frac{1}{V} \int f(O, \Gamma^\prime) d_{\mu}(O) &=& \frac{1}{V} \int f(OO^\prime, \Gamma) d_{\mu}(O) \nonumber \\
      &=& \frac{1}{V} \int f(O, \Gamma) d_{\mu}(O).
  \end{eqnarray}
  The above relation renders the averaging over $\{\bm{\Gamma}\}$ in Eq. \eqref{eq:hagen} redundant. Therefore, we finally have
  \begin{eqnarray}
       \mathbb{E}[f] = \frac{1}{V} \int f(O, \Gamma) d_{\mu}(O),
  \end{eqnarray}
  for any $\Gamma \in \{\bm{\Gamma}\}$. Hence the proof.
  \end{proof}

\section{Typical Genuine entanglement of pure Gaussian states with a large number of modes}
\label{sec:nh}
For a random pure Gaussian state with a large number of modes, a lot of interesting features emerge. The most prominent being the concentration of the value of GGM around its Haar averaged value. We elaborate on these features in this section. Regarding Haar uniform generation of Gaussian states, following Theorem. \ref{theorem:1}, we choose $\Gamma = D \oplus D^{-1}$, where $D = \frac{\bar{\nu}}{n}\mathbb{I}$ without loss of any generality.

\begin{lemma}
    For a random $n$-mode $(n \! >\!> \!1)$ pure Gaussian state, the maximal Schmidt coefficient comes from the single mode reduced sector.
    \label{lemma:1}
\end{lemma}
\begin{proof}
 The starting point of the proof is a result from \cite{fukuda} where they show for any $k$-mode reduced subsystem, where the maximal subsystem size $k_{\max} \leq K n^\alpha$ with $0 \leq \alpha < 1$, we have 
\begin{eqnarray}
    \nu_i^2 = \bar{\nu}^2 - \mathcal{O}\Big(\frac{1}{n^{1-\alpha}}\Big).
\end{eqnarray}
Now for any large enough $n \in \mathbb{N}, ~\exists ~\kappa < 1, \text{ s.t. } k_{\max} = [\frac{n}{2}].$ Here $[x]$ represents the integral part of $x$.
 Therefore, for for $n \! >\!\!> \!1$, the symplectic eigenvalues $\{\nu_i\}_{i = 1}^k$ of any $k$-mode reduced covariance matrix ($k \leq k_{\max} = [\frac{n}{2}]$) of a random $n$-mode pure Gaussian state are almost identically equal to the average energy per mode, $\bar{\nu}$. 
 \begin{eqnarray}
     \nu_i \approx \bar{\nu} ~\forall i \in [1,k].
     \label{eq:sympletic1}
 \end{eqnarray}
Note that $\{\nu_i\}_{i = 1}^k$ constitutes a product of $k$ thermal states, whose maximal eigenvalue is simply the product of the largest eigenvalues of each constituent thermal state with inverse temperature $\beta_i = \ln \frac{\nu_i + 1}{\nu_i - 1}$. Since the eigenvalues of the tensor product of states multiply, the largest Schmidt coefficient from any $k$-mode reduction is the product of the largest eigenvalue of each thermal state. Recalling that the largest eigenvalue of a thermal state is its first one, the largest eigenvalue of any $k$-mode reduction is
\begin{eqnarray}
   \lambda_{\max}^k = \prod_{i = 1}^k (1 - e^{-\beta_i}) = \prod_{i = 1}^k \frac{1}{1 + \nu_i} \approx \Big(\frac{2}{1 + \bar{\nu}}\Big)^k.
    \end{eqnarray}
Since $\bar{\nu} \geq 1$, we have
 \begin{eqnarray}
     \frac{2}{1 + \bar{\nu}} > \Big(\frac{2}{1 + \bar{\nu}}\Big)^k ~\forall ~~~k > 1.
 \end{eqnarray}
Therefore the largest eigenvalue $\lambda_{\max}$ comes from the single-mode sector, with
\begin{eqnarray}
    \lambda_{\max} := \max_k \{\lambda_{\max}^k\} = \lambda_{\max}^{k =1}  \approx \frac{2}{1+ \bar{\nu}}.
    \label{eq:lambdamax}
\end{eqnarray}
Hence the proof.
\end{proof}
\noindent With this, we are ready to state the following theorem.
\begin{theorem}
    The typical GGM of a random $n$-mode pure Gaussian states for $n\!>\!\!>\!1$ is given by
    \begin{eqnarray}
        \overline{\mathcal{G}} \approx \frac{\bar{\nu}-1}{\bar{\nu}+1}.
    \end{eqnarray}
    Here,  we have
    $\mathbb{E}[\mathcal{G}_{n >\!\!>1}] := \overline{\mathcal{G}}$.
    \label{theorem:2}
\end{theorem}
\begin{proof}
    From Lemma. \ref{lemma:1}, we know the largest eigenvalue comes from the single-mode sector and is given by $\lambda_{\max} \approx \frac{2}{1 + \bar{\nu}}$, see Eq. \eqref{eq:lambdamax}. The typical value of the GGM then is given by
    \begin{eqnarray}
       \overline{\mathcal{G}} = 1 - \lambda_{\max} \approx \frac{\bar{\nu}-1}{\bar{\nu}+1}. 
    \end{eqnarray}
     Hence the proof.
   \end{proof}

\noindent Finally, we attempt to bound the deviation of GGM of a typical state from its Haar averaged value. 

\begin{theorem}
    For $n>\!>1$, the GGM $\mathcal{G}$ of any random Gaussian state satisfies
    \begin{eqnarray}
{\tt Prob}\{ (\mathcal{G} - \overline{\mathcal{G}})^2 > \epsilon \} \leq \exp (-c \epsilon^2 n),
    \end{eqnarray}
    where $\epsilon > \frac{C}{n}$. Here  $C,c>0$ are universal constants. 
\end{theorem}
\begin{proof}
    For $n>\!>1$, it follows from Lemma. \ref{lemma:1} that the largest eigenvalue comes from the single mode sector. Let the largest eigenvalue be induced by the symplectic eigenvalue $\nu$, which we know from Eq. \eqref{eq:sympletic1} that $\nu \approx \bar{\nu}$. The corresponding GGM is
    \begin{eqnarray}
        \mathcal{G} = 1 - \frac{2}{1 + \nu}.
    \end{eqnarray}
    We now intend to examine its deviation from the mean value $\overline{\mathcal{G}}$. Therefore, the  quantity of interest is $(\mathcal{G} - \overline{\mathcal{G}})^2$, where the typical value of GGM,  $\overline{\mathcal{G}}$, is obtained in Theorem. \eqref{theorem:1}.   
    
    On the other, from \cite{fukuda}, for universal constants $C,c > 0$, and $\epsilon > \frac{C}{n}$ 
    \begin{eqnarray}
       {\tt Prob}\{ (\nu^2 - \bar{\nu}^2)^2 > \epsilon \} \leq \exp (-c \epsilon^2 n). 
       \label{eq:fukuda}
    \end{eqnarray}
   The proof proceeds by noting that 
    \begin{eqnarray}
     \overline{\mathcal{G}} - \mathcal{G} = \frac{2(\nu - \bar{\nu})}{(1+\nu)(1+\bar{\nu})}.
    \end{eqnarray}
    This in turn translates to
    \begin{eqnarray}
        \nu^2 - \bar{\nu}^2 = \frac{1}{2}(\overline{\mathcal{G}} - \mathcal{G})(1+\nu)(1+\bar{\nu})(\nu + \bar{\nu}).
    \end{eqnarray}
    Since $\nu, \bar{\nu} \geq 1$, we have
    \begin{eqnarray}
    (\nu^2 - \bar{\nu}^2)^2 \geq   (\mathcal{G} - \overline{\mathcal{G}})^2.  
    \label{eq:nuG}
    \end{eqnarray}
    Finally, by combining Eqs. \eqref{eq:fukuda} and \eqref{eq:nuG} we arrive at
     \begin{eqnarray}
        {\tt Prob}\{ (\mathcal{G} - \overline{\mathcal{G}})^2 > \epsilon \} \leq \exp (-c \epsilon^2 n).
    \end{eqnarray}
    Hence the proof.
\end{proof}

\section{Typical Genuine entanglement of pure Gaussian states with a few modes}
\label{sec:nl}
In the previous section, we demonstrated that the genuine entanglement content of a typical state is sharply concentrated around its mean value for Gaussian states with a large number of modes. This section intends to visualize this phenomenon with typical Gaussian states with finite modes.

To this goal, we randomly generate $N$
pure multimode Gaussian states of three, four, five and six modes and compute their GGM. A brief description of the procedure for generating random pure Gaussian states follows.
From Eq. \eqref{eq:random1}, the covariance matrix of any pure Gaussian state is of the form $\sigma = O \Gamma O^T$.
To construct a random $\sigma$, we first need to
choose a $\Gamma$. Following Theorem. \ref{theorem:1}, any $\Gamma$ consistent with Eq. \eqref{eq:enboundpm2} is equivalent. Therefore, without loss of any generality we choose $\Gamma = D \oplus D^{-1}$, where $D = \frac{\bar{\nu}}{n}\mathbb{I}$.
Now, we sample a random unitary matrix $U$ following the Haar measure in $U(n)$. See \cite{Karol,Karol2} for a detailed description of Haar uniform generation of $U \in U(n)$. 
We now briefly describe their generation procedure. First, we generate $n$ random pure states $\{\ket{\psi_1}, \ket{\psi_2}, \ldots, \ket{\psi_n} \}$ each of ${\tt dim} = n$. 
Here $\ket{\psi_k} = {\tt normalize} [(a_k^1 + i b_k^1, a_k^2 + i b_k^2, \ldots, a_k^n+ i b_k^n)^T]$, where all $(a_k^j,b_k^j)$s are independently chosen from a standard normal distribution.
By a standard algorithm, Grahm-Schmidt orthogonalization \cite{Arfken}, we arrive at an orthonormal set $\{\ket{\phi_1}, \ket{\phi_2}, \ldots, \ket{\phi_n} \}= {\tt orthoganalize}[\{\ket{\psi_1}, \ket{\psi_2}, \ldots, \ket{\psi_n} \}]$.
The unitary can be constructed as
\begin{eqnarray}
    U = \begin{bmatrix}
       . & & \ldots & \\
        & & \ldots & \\
        \ket{\phi_1} & \ket{\phi_2} & \ldots & \ket{\phi_n} \\
        & & \ldots & \\
        & & \ldots & 
    \end{bmatrix}.
\end{eqnarray}

Finally, we obtain the required orthogonal symplectic matrix $O$ from $U$ from Eq. \eqref{eq:u2o}.

\begin{figure}[ht]
    \centering
    \includegraphics[width=\linewidth]{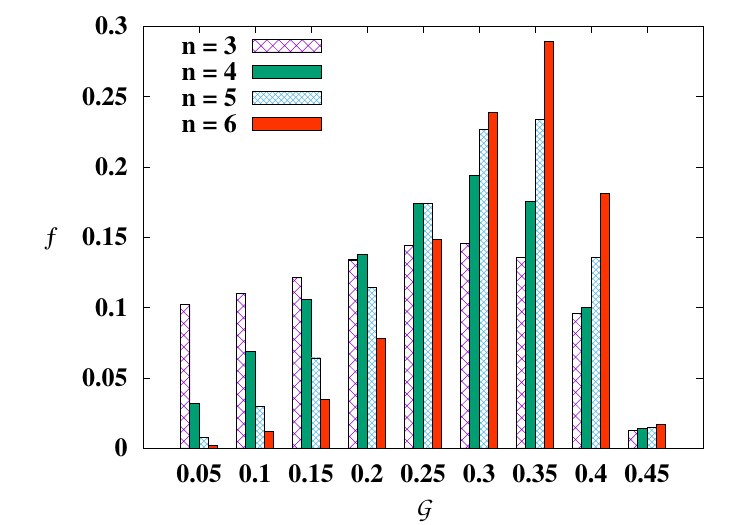}
    \caption{The bars at a given value $x$ represents the fraction of randomly generated $n$-mode pure Gaussian states $f$ with GGM values between $x - 0.05$ to $x$. The distribution shifts progressively towards higher GGM values with an increase in the number of modes. The average energy per mode is chosen to be $\bar{\nu} = 2.6$. Both axes are dimensionless.}
    \label{fig:fig1}
\end{figure}

The numerical experiment is performed by generating $N \sim 10^6$ $n$-mode pure Gaussian state with an average energy per mode fixed to $\bar{\nu} = 2.6$. Using Eq. \eqref{eq:GGM-new-form}, we compute and record the GGM of each such generated state. 
We use this data to construct the distribution of GGM values across the range consistent with the chosen average energy per mode, see Fig. \ref{fig:fig1}. We find that with an increase in the number of modes, the distribution progressively sharpens from a relatively flat one that we get for $n = 3$. On top of the sharpening (concentration), we notice a progressive shift of the overall distribution towards higher GGM values with an increasing number of modes.

We then move on to calculate statistical averages like, mean, the Haar averaged GGM $\mathbb{E}[\mathcal{G}_n]$ and the corresponding standard deviation $\mathbb{E}[\Delta \mathcal{G}_n]$ and track them with the increase in the number of modes:
\begin{eqnarray}
  \mathbb{E}[\mathcal{G}_n] &=& \frac{1}{N} \sum_{i = 1}^N \mathcal{G}^i_n, \nonumber \\
  \mathbb{E}[\Delta \mathcal{G}_n] &=& \sqrt{\mathbb{E}[\mathcal{G}_n^2] - \mathbb{E}^2[\mathcal{G}_n]},
\end{eqnarray}
where $\mathcal{G}^i_n$ is the GGM of the $i^{\text{th}}$ randomly generated $n$-mode pure Gaussian state, and $\mathbb{E}[\mathcal{G}_n^2] = \frac{1}{N} \sum_{i = 1}^N (\mathcal{G}^i_n)^2$.
 We find that the average value of GGM increases as the number of modes increases. At the same time, the standard deviation of the distribution of Haar uniformly generated $n$-mode pure Gaussian states decrease monotonically, see Table. \ref{tab:my_label} for exact values.
\begin{table}[h]
    \centering
    \begin{tabular}{|c|c|c|c|c|}\hline
        $n$ & $3$ & $4$ & $5$ & $6$  \\ \hline
     $\mathbb{E}[\mathcal{G}_n]$     & $~0.2068~$ & $~0.2357~$ & $~0.2647~$ & $~0.2874~$\\ \hline
    $\mathbb{E}[\Delta \mathcal{G}_n]$     & $0.1101$ & $0.0941$ & $0.0814$  & $0.0707$ \\ \hline
    \end{tabular}
    \caption{The Haar averaged value of GGM   $\mathbb{E}[\mathcal{G}_n]$ and the standard deviation   $\mathbb{E}[\Delta\mathcal{G}_n]$ for $n = 3$ to $6$. The average energy per mode is restricted to $\bar{\nu} = 2.6$.}
    \label{tab:my_label}
\end{table}
As mentioned before, in our case, we choose $N = 10^6$. The convergence of the reported numbers for $N = 10^6$ is guaranteed by the invariance of these values by changing the number of generated states. The result remains qualitatively similar for other choices of average energy per mode as well. Finally, for $\bar{\nu} = 2.6$ and $n >\!> 1$, Theorem. \ref{theorem:2} predicts $\mathbb{E}[\mathcal{G}_{n >\!> 1}]=\overline{\mathcal{G}} = \frac{4}{9} \approx 0.4444$. We verify this numerically for $n = 50.$ Assuming Lemma. \ref{lemma:1} holds for $n = 50$, we get
$ \mathbb{E}[\mathcal{G}_{n = 50}] \approx 0.4129$ which is pretty close to $\overline{\mathcal{G}}$. The corresponding distribution is also quite sharp as indicated by a minuscule standard deviation: $
  \mathbb{E}[\Delta \mathcal{G}_{n=50}] \approx 0.0089.$

\section{conclusion}
\label{sec:con}
Investigation of typical properties of a certain set of states has attained a lot of focus, particularly in finite dimensional systems \cite{typical1,typical2,typical3,typical4,t5,t6,t7}. In the continuous variable sector as well, typical features of Haar uniform Gaussian states ranging from bipartite entanglement and work extraction have been examined in \cite{fukuda,uttam}. 

In this work, we first outline the generation of Haar uniform multimode pure Gaussian states and derive an expression for the Haar averaged value of any function on the set of generated states. 
Then we move on to the central tenant of our work: investigation of the genuine entanglement characteristics of Haar uniformly generated multimode pure Gaussian states with a fixed average energy per mode. We choose the generalized geometric measure (GGM) \cite{ggm1, ggmcv} to quantify the genuine entanglement content of the Haar uniformly generated states.
We accomplish this aim by initially focusing on states with a large number of modes. By leveraging the concentration of the symplectic spectrum \cite{fukuda}, we infer that the generalized geometric measure (GGM) content can be described by the symplectic eigenvalues of single-mode reductions. This insight facilitates a closed-form expression for the Haar averaged GGM based on the average energy per mode. Additionally, we show that, for typical Haar states, the GGM closely resembles its Haar averaged value. We quantify this closeness by demonstrating the probability of the GGM of a typical state to be different from the Haar averaged value is upper bounded with an exponentially suppressed bound. Note that the above findings hold for states with a large number of modes.

Finally, we investigate the patterns of GGM distribution of Haar uniformly generated multimode pure Gaussian states with three, four, five and six modes.
 Our findings show that the distribution becomes more concentrated and shifts towards larger GGM values as the number of modes grows. By calculating the Haar averaged GGM and the GGM distribution standard deviation, we are able to quantify these patterns and demonstrate that, as the number of modes rises, the former increases while the latter decreases.
Overall, we believe that our work sheds some light on the genuine entanglement properties of typical multimode Gaussian states. We hope that the findings presented in this work will serve as a foundation for future research in this direction.

\section*{Acknowledgement}
This work is supported by the Hong Kong Research Grant
Council through Grants No. 17307719 and 17307520 and
though the Senior Research Fellowship Scheme SRFS2021-7S02, by the Croucher Foundation, and by the John Templeton
Foundation through Grant No. 62312, The Quantum Information Structure of Space-time (qiss.fr). The opinions expressed in this publication are those of the authors and
do not necessarily reflect the views of the John Templeton Foundation.

\bibliography{bib}

\end{document}